\documentclass[11pt]{article}
\usepackage[utf8]{inputenc}
\usepackage{fullpage}
\usepackage{mathtools}
\usepackage{breqn}
\usepackage{amssymb,amsmath,amsthm}
\usepackage{epsfig}
\usepackage{url}
\usepackage{color}
\usepackage{array}
\usepackage[nocomma]{optidef}
\usepackage[ruled,linesnumbered]{algorithm2e}

\newtheorem{thm}{Theorem}%[section]

\newtheorem{lem}[thm]{Lemma}

\theoremstyle{definition}

\newtheorem{defn}[thm]{Definition}

\newtheorem*{example*}{Example}

\begin{document}
\title{Minimal Envy Matchings in the Hospitals/Residents Problem with Lower Quotas}
\author{Changyong Hu\thanks{Department of Electrical and Computer Engineering, University of Texas at Austin, Austin, Texas 78705, USA. E-mail: \texttt{colinhu9@utexas.edu.}} \and Vijay K. Garg\thanks{Department of Electrical and Computer Engineering, University of Texas at Austin, Austin, Texas 78705, USA. E-mail: \texttt{garg@ece.utexas.edu}.}
%Research supported by the Lend\"ulet program of the Hungarian Academy of Sciences (MTA), under grant number LP2017-19/2017
}
%\date{}
\maketitle

\begin{abstract}
In the Hospitals/Residents problem, every hospital has an upper quota that limits the number of residents assigned to it. While, in some applications, each hospital also has a lower quota for the number of residents it receives. In this setting, a stable matching may not exist. Envy-freeness is introduced as a relaxation of stability that allows blocking pairs involving a resident and an empty position of a hospital. While, envy-free matching might not exist either when lower quotas are introduced. We consider the problem of finding a feasible matching that satisfies lower quotas and upper quotas and minimizes envy in terms of envy-pairs and envy-residents in the Hospitals/Resident problem with Lower Quota. We show that the problem is NP-hard with both envy measurement. We also give a simple exponential-time algorithm for the Minimum-Envy-Pair HRLQ problem.

\end{abstract}

\section{Introduction}
In the Hospitals/Residents problem, every hospital has an upper quota that limits the number of residents assigned to it. While, in some applications, each hospital also has a lower quota for the number of residents it receives. This extension of the HR problem is referred to as the Hospitals/Residents problem with Lower Quotas (HRLQ, for short) in the literature. However, the existence of a stable matching is not always true for a given HRLQ instance. This can be easily observed by the well-known Rural Hospitals theorem \cite{gale1985some, roth1984evolution, roth1986allocation} stating that each hospital is assigned the same number of residents in any stable matching for a given HR instance.

Since there might be no stable matching given a HRLQ instance, one natural approach is to weaken the requirement of stability. Envy-freeness is a relaxation of stability that allows blocking pairs involving a resident and an empty position of a hospital. Structural results of envy-free matchings were investigated in \cite{wu2018lattice} showing that the set of envy-free matchings forms a lattice. Envy-free matchings with lower quotas has recently been studied in \cite{yokoi2020envy}. Fragiadalis et al. \cite{fragiadakis2016strategyproof} studies envy-free matchings (called fairness in their papar) when the preference lists are complete and the sum of lower quotas of all hospitals does not exceed the number of residents. In this restricted setting, envy-free matching always exists and they gave a linear-time algorithm called extended-seat deferred acceptance (ESDA) to find an envy-free matching. However, if the preference lists are not complete, envy-free matching may not exist. Yokoi \cite{yokoi2020envy} provided a characterization of envy-free matchings in HRLQ, connecting them to stable matchings in a modified HR instance so that the existence of envy-free matching can be decided by running DA on this modified HR instance. 

Given a HRLQ instance, we know that stable matchings or envy-free matchings might not exist. Hamada et al. \cite{hamada2016hospitals} considered the problem of minimizing the number of blocking pairs among all feasible matchings (feasible matching means a matching satisfies both lower and upper quotas of each hospital). They showed hardness of approximation of the problem and provided an exponential-time exact algorithm. 

In this paper, we consider the problem of minimizing envy (defined later) among all feasible matchings given a HRLQ instance. Given a feasible matching, the envy among all residents can be measured by the number of envy-pairs or the number of residents involved in some envy-pairs.

We show that for both measurements, the problem is NP-hard. We also provide an exponential-time algorithm to find a feasible matching minimizing the number of envy-pairs.

\subsection{Related Works} Popular matching is another relaxation of stability which preserves ``global" stability in the sense that no majority of residents wish to alter into another feasible matching. Popular matching always exists. Nasre and Nimbhorkar \cite{nasre2017popular} proposed an efficient algorithm to compute a maximum cardinality matching that is popular amongst all the feasible matchings in an HRLQ instance. When there exists a stable matching given a HRLQ instance, it is known that every stable matching is popular. 
\section{Preliminaries}
An instance $I = (G, \succ)$ of the Hospitals/Residents Problem with Lower Quotas (HRLQ for short) consists of a bipartite graph $G = (R \cup H, E)$, where $R$ is a set of residents and $H$ is a set of hospitals, and an edge $(r,h) \in E$ denotes that $r$ and $h$ are mutually acceptable, and a preference system $\succ$ such that every vertex (resident and hospital) in $G$ ranks its neighbors in a strict order, referred as the preference list of the vertex. If a vertex $a$ prefers its neighbor $b_1$ over $b_2$, we denote it by $b_1 \succ_a b_2$. Each hospital $h$ has a lower quota $l_h$ and an upper quota $u_h$ ($0 \leq l_h \leq u_h$). Sometimes we write $[l_h, u_h]$ to denote the lower and upper quota for some hospital $h$.

A matching $M \subseteq E$ in $G$ is an assignment of residents to hospitals such that each resident is matched to at most one hospital, and every hospital $h$ is matched to at most $u_h$ residents. Let $M(r)$ denote the hospital that $r$ is matched in $M$. If $r$ is unmatched in $M$, we let $M(r) = \emptyset$. For any neighbor $h$ of $r$, we have $h \succ_r \emptyset$ since we assume that any resident prefers to be matched over to be unmatched. We say that a hospital $h$ is under-subscribed in $M$ if $|M(h)| < u_h$, $h$ is fully-subscribed if $|M(h)| = u_h$, $h$ is over-subscribed if $|M(h)| > u_h$ and $h$ is deficient if $|M(h)| < l_h$. A matching $M$ is feasible in a HRLQ instance if no hospital is deficient or over-subscribed in $M$. The HRLQ problem is to match residents to hospitals under some optimality condition such that the matching is feasible. Envy-freeness is defined as follows.

\begin{defn}
Given a matching $M$, a resident-hospital pair $(r,h)$ is an {\em envy-pair} if $h \succ_r M(r)$ and $r \succ_h r'$ for some $r' \in M(h)$.
\end{defn}

\begin{defn}
Given a matching $M$, a resident $r$ has {\em justified envy} toward $r'$ who is matched to hospital $h$ if $h \succ_r M(r)$ and $r \succ_h r'$.
\end{defn}

\begin{defn}
A matching $M$ is {\em envy-free} if there is no {\em envy-pair} in $M$. Equivalently, a matching $M$ is envy-free if no resident has {\em justified envy} toward other residents in $M$.
\end{defn}

In case that envy-free matchings may not exist in a given HRLQ instance. We define two other problems that minimizes envy in terms of envy-pairs and envy-residents.

\textbf{Minimum-Envy-Pair Hospitals/Residents Problem with Lower Quotas} (Min-EP HRLQ for short) is the problem of finding a feasible matching with the minimum number of envy-pairs. 0-1 Min-EP HRLQ is the restriction of Min-EP HRLQ where a quota of each hospital is either $[0,1]$ or $[1,1]$. 

\begin{defn}
Given a matching $M$, a resident $r$ is an envy-resident if there exists $h$ such that $(r,h) \in E$ is an envy-pair in $M$.
\end{defn}

\textbf{Minimum-Envy-Resident Hospitals/Residents Problem with Lower Quotas} (Min-ER HRLQ for short) is the problem of finding a feasible matching with minimum number of envy-residents. 0-1 Min-ER HRLQ is defined similarly.

We assume without loss of generality that the number of residents is at least the sum of the lower quotas of all hospitals, since otherwise there is no feasible matching. In other papers, they impose the Complete List restriction (CL-restriction for short). There always exists envy-free matchings in CL-restriction instances, and a maximum-size one can be found in polynomial time \cite{fragiadakis2016strategyproof, krishnaa2020envy}.

\section{Minimum-Envy-Pair HRLQ}
In this section, we consider the problem of minimizing the number of envy-pairs in HRLQ. We prove a NP-hardness result for 0-1 Min-EP HRLQ in the following theorem.

\begin{thm}
0-1 Min-EP HRLQ is NP-hard. 
\end{thm}
\begin{proof}
We give a polynomial-time reduction from the well-known NP-complete problem Vertex Cover. Below is the definition of the decision version of the Vertex Cover problem. Given a graph $G=(V,E)$ and a positive integer $K \leq |V|$, we are asked if there is a subset $C \subseteq V$ such that $|C| \leq K$, which contains at least one endpoint of each edge of $G$. 

{\bf Reduction:} Given a graph $G=(V,E)$ and a positive integer $K$, which is an instance of the Vertex Cover problem, we construct an instance $I$ of 0-1 Min-EP HRLQ. Define $n = |V|$, $m = |E|$ and $l = n^2 + 1$. The set of residents is $R = C \cup F \cup S$ and the set of hospitals is $H = V \cup T$. Each set is defined as follows:

\begin{align*}
    C &= \{c_i ~|~ 1 \leq i \leq K\} \\
    F &= \{f_i ~|~ 1 \leq i \leq n - K\} \\
    S^{i,j} &= \{s^{i,j}_{0,a} ~|~ 1 \leq a \leq l\} \cup \{s^{i,j}_{1,a} ~|~ 1 \leq a \leq l\} ~ ((v_i, v_j) \in E, i < j)\\
    S &= \bigcup S^{i,j}\\
    V &= \{v_i ~|~ 1 \leq i \leq n\}\\
    T^{i,j} &= \{t^{i,j}_{0,a} ~|~ 1 \leq a \leq l\} \cup \{t^{i,j}_{1,a} ~|~ 1 \leq a \leq l\} ~ ((v_i, v_j) \in E, i < j)\\
    T &= \bigcup T^{i,j}
\end{align*}

Each hospital in $H = V \cup T$ has a quota $[1,1]$. Note that $|C| + |F| = |V| = n$ and $|S| = |T| = 2ml$. Thus $|H| + |R| = 2n+4ml$, which is polynomial in $n$ and $m$.

Next, we construct the preference lists of $I$. The preference lists of residents is shown as follows:
\begin{align*}
    c_i &: [[V]] && (1 \leq i \leq K)\\
    f_i &: [[V]] && (1 \leq i \leq n-K)\\
    s^{i,j}_{0,1} &: t^{i,j}_{0,1} ~~~ v_i ~~~ t^{i,j}_{1,1} && ((v_i, v_j) \in E, i < j)\\
    s^{i,j}_{0,2} &: t^{i,j}_{0,2} ~~~ v_i ~~~ t^{i,j}_{0,3} && ((v_i, v_j) \in E, i < j)\\
    &~~~~~\vdots\\
    s^{i,j}_{0,l-1} &: t^{i,j}_{0,l-1} ~~~ v_i ~~~ t^{i,j}_{0,l} && ((v_i, v_j) \in E, i < j)\\
    s^{i,j}_{0,l} &: t^{i,j}_{0,l} ~~~ v_i ~~~ t^{i,j}_{0,1} && ((v_i, v_j) \in E, i < j)\\
    s^{i,j}_{1,1} &: t^{i,j}_{0,2} ~~~ v_j ~~~ t^{i,j}_{1,2} && ((v_i, v_j) \in E, i < j)\\
    s^{i,j}_{1,2} &: t^{i,j}_{1,2} ~~~ v_j ~~~ t^{i,j}_{1,3} && ((v_i, v_j) \in E, i < j)\\
    &~~~~~\vdots\\
    s^{i,j}_{1,l-1} &: t^{i,j}_{1,l-1} ~~~ v_j ~~~ t^{i,j}_{1,l} && ((v_i, v_j) \in E, i < j)\\
    s^{i,j}_{1,l} &: t^{i,j}_{1,l} ~~~ v_j ~~~ t^{i,j}_{1,1} && ((v_i, v_j) \in E, i < j)
\end{align*}
, where $[[V]]$ denotes a fixed order of elements in $V$ in an increasing order of indices.

The preference lists of hospitals is shown as follows:
\begin{align*}
    v_i &: [[C]] ~~ [[S_i]] ~~ [[F]] && (1 \leq i \leq n)\\
    t^{i,j}_{0,1} &: s^{i,j}_{0,1} ~~~ s^{i,j}_{0,l} && ((v_i, v_j) \in E, i < j)\\
    t^{i,j}_{0,2} &: s^{i,j}_{1,1} ~~~ s^{i,j}_{0,2} && ((v_i, v_j) \in E, i < j)\\
    &~~~~~\vdots\\
    t^{i,j}_{0,l-1} &: s^{i,j}_{0,l-2} ~~~ s^{i,j}_{0,l-1} && ((v_i, v_j) \in E, i < j)\\
    t^{i,j}_{0,l} &: s^{i,j}_{0,l-1} ~~~ s^{i,j}_{0,l} && ((v_i, v_j) \in E, i < j)\\
    t^{i,j}_{1,1} &: s^{i,j}_{0,1} ~~~ s^{i,j}_{1,l} && ((v_i, v_j) \in E, i < j)\\
    t^{i,j}_{1,2} &: s^{i,j}_{1,1} ~~~ s^{i,j}_{1,2} && ((v_i, v_j) \in E, i < j)\\
    &~~~~~\vdots\\
    t^{i,j}_{1,l-1} &: s^{i,j}_{1,l-2} ~~~ s^{i,j}_{1,l-1} && ((v_i, v_j) \in E, i < j)\\
    t^{i,j}_{1,l} &: s^{i,j}_{1,l-1} ~~~ s^{i,j}_{1,l} && ((v_i, v_j) \in E, i < j)
\end{align*}
, where $[[C]]$ and $[[F]]$ are as before a fixed order of all the residents in $C$ and $F$, respectively, in an increasing order of indices, $[[S_i]]$ is an arbitrary order of all the residents in $S$ that is acceptable to $v_i$ as determined by the preference lists of residents.

Since each hospital in $V \cup T$ has a quota $[1,1]$, any feasible matching must be a perfect matching in $I$. Further, the subgraph between $C \cup F$ and $V$ is a complete graph and there is no edge between $C \cup F$ and $T$ in $I$. Thus, any feasible matching must include a perfect matching between $C \cup F$ and $V$.

For each edge $(v_i, v_j) \in E (i < j)$, there are the set of residents $S^{i,j}$ and the set of hospitals $T^{i,j}$. We call this pair of sets a $g_{i,j}$-gadget, and write it as $g_{i,j} = (S^{i,j}, T^{i,j})$. For each gadget $g_{i,j}$, let us define two perfect matchings between $S^{i,j}$ and $T^{i,j}$ as follows:

\begin{align*}
    M^{i,j}_0 &= \{(s^{i,j}_{0,1}, t^{i,j}_{0,1}), (s^{i,j}_{0,2}, t^{i,j}_{0,2}), \cdots, (s^{i,j}_{0,l-1}, t^{i,j}_{0,l-1}), (s^{i,j}_{0,l}, t^{i,j}_{0,l}),\\
    &(s^{i,j}_{1,1}, t^{i,j}_{1,2}), (s^{i,j}_{1,2}, t^{i,j}_{1,3}), \cdots, (s^{i,j}_{1,l-1}, t^{i,j}_{1,l}), (s^{i,j}_{1,l}, t^{i,j}_{1,1})\}, and\\
    M^{i,j}_1 &= \{(s^{i,j}_{0,1}, t^{i,j}_{1,1}), (s^{i,j}_{0,2}, t^{i,j}_{0,3}), \cdots, (s^{i,j}_{0,l-1}, t^{i,j}_{0,l}), (s^{i,j}_{0,l}, t^{i,j}_{0,1}),\\
    &(s^{i,j}_{1,1}, t^{i,j}_{0,2}), (s^{i,j}_{1,2}, t^{i,j}_{1,2}), \cdots, (s^{i,j}_{1,l-1}, t^{i,j}_{1,l-1}), (s^{i,j}_{1,l}, t^{i,j}_{1,l})\}
\end{align*}

Figure \ref{fig:perfect} shows $M^{i,j}_0$ and $M^{i,j}_1$ on preference lists of $S^{i,j}$.
\begin{figure}[ht!]
    \centering
\begin{tabular}{ m{1cm} m{1.4cm} m{0.9cm} m{2.5cm} m{1cm} m{1.4cm} m{0.9cm} m{1cm}}
 $s^{i,j}_{0,1}$ & $:~\underline{t^{i,j}_{0,1}}$ & $v_i$ & $t^{i,j}_{1,1}$ & $s^{i,j}_{0,1}$ & $:~ t^{i,j}_{0,1}$ & $v_i$ & $\underline{t^{i,j}_{1,1}}$\\
 $s^{i,j}_{0,2}$ & $:~\underline{t^{i,j}_{0,2}}$ & $v_i$ & $t^{i,j}_{0,3}$ & $s^{i,j}_{0,2}$ & $:~ t^{i,j}_{0,2}$ & $v_i$ & $\underline{t^{i,j}_{0,3}}$\\
 &&$\vdots$&&&&$\vdots$\\
 $s^{i,j}_{0,l-1}$ & $:~\underline{t^{i,j}_{0,l-1}}$ & $v_i$ & $t^{i,j}_{0,l}$ & $s^{i,j}_{0,l-1}$ & $:~ t^{i,j}_{0,l-1}$ & $v_i$ & $\underline{t^{i,j}_{0,l}}$\\
 $s^{i,j}_{0,l}$ & $:~\underline{t^{i,j}_{0,l}}$ & $v_i$ & $t^{i,j}_{0,1}$ & $s^{i,j}_{0,l}$ & $:~ t^{i,j}_{0,l}$ & $v_i$ & $\underline{t^{i,j}_{0,1}}$\\
 $s^{i,j}_{1,1}$ & $:~t^{i,j}_{0,2}$ & $v_j$ & $\underline{t^{i,j}_{1,2}}$ & $s^{i,j}_{1,1}$ & $:~ \underline{t^{i,j}_{0,2}}$ & $v_j$ & $t^{i,j}_{1,2}$\\
 $s^{i,j}_{1,2}$ & $:~t^{i,j}_{1,2}$ & $v_j$ & $\underline{t^{i,j}_{1,3}}$ & $s^{i,j}_{1,2}$ & $:~ \underline{t^{i,j}_{1,2}}$ & $v_j$ & $t^{i,j}_{1,3}$\\
 &&$\vdots$&&&&$\vdots$\\
 $s^{i,j}_{1,l-1}$ & $:~t^{i,j}_{1,l-1}$ & $v_j$ & $\underline{t^{i,j}_{1,l}}$ & $s^{i,j}_{1,l-1}$ & $:~ \underline{t^{i,j}_{1,l-1}}$ & $v_j$ & $t^{i,j}_{1,l}$\\
 $s^{i,j}_{1,l}$ & $:~t^{i,j}_{1,l}$ & $v_j$ & $\underline{t^{i,j}_{1,1}}$ & $s^{i,j}_{1,l}$ & $:~ \underline{t^{i,j}_{1,l}}$ & $v_j$ & $t^{i,j}_{1,1}$
\end{tabular}
    \caption{Matchings $M^{i,j}_0$ (left) and $M^{i,j}_1$ (right)}
    \label{fig:perfect}
\end{figure}

\begin{lem}\label{lem:perfect}
For a gadget $g_{i,j} = (S^{i,j}, T^{i,j})$, $M^{i,j}_0$ and $M^{i,j}_1$ are the only perfect matchings between $S^{i,j}$ and $T^{i,j}$. Furthermore, each $M^{i,j}_0$ and $M^{i,j}_1$ contains an unique envy-pair $(r,h)$ such that $r \in S^{i,j}$ and $h \in T^{i,j}$.
\end{lem}
\begin{proof}
Consider the induced subgraph $G^{i,j}$ that contains only $S^{i,j}$ and $T^{i,j}$. One can see that $G^{i,j}$ is a cycle of length $4l$. Hence there are only two perfect matchings between $S^{i,j}$ and $T^{i,j}$, and they are actually $M^{i,j}_0$ and $M^{i,j}_1$. Moreover, it is easy to check that $M^{i,j}_0$ contains only one envy-pair $(s^{i,j}_{1,1}, t^{i,j}_{0,2})$ and $M^{i,j_1}$ contains only one envy-pair $(s^{i,j}_{0,1}, t^{i,j}_{0,1})$.
\end{proof}

We now ready to show the NP-hardness of 0-1 Min-EP HRLQ.

\begin{lem}\label{lem:yes-ep}
If $G$ is a ``yes" instance of the Vertex Cover problem, then $I$ has a solution with at most $n^2 + m$ envy-pairs.
\end{lem}
\begin{proof}
If $G$ is a ``yes" instance of the Vertex Cover problem, then $G$ has a vertex cover of size exactly $K$. Let this vertex cover be $V_c \subseteq V$ and let $V_f = V \backslash V_c$. We construct a matching $M$ of $I$ according to $V_c$. We match each hospital in $V_c$ to each resident in $C$ and each hospital in $V_f$ to each resident in $F$ in an arbitrary way. Since $|C \cup F| = |V| = n$, there are at most $n^2$ envy-pairs between $C \cup F$ and $V$. 

Since $V_c$ is a vertex cover of $G$, for each edge $(v_i, v_j)$, either $v_i$ or $v_j$ is included in $V_c$. Thus for each gadget $g_{i,j} = (S^{i,j}, T^{i,j})$ corresponding to the edge $(v_i, v_j)$, if $v_i \in V_c$, we use $M^{i,j}_1$ as part of our matching $M$, otherwise, use $M^{i,j}_0$. By constructing in this way, it is easy to see that neither $v_i$ nor $v_j$ can be involved in an envy-pair. Also, by Lemma \ref{lem:perfect}, there is only one envy-pair between $S^{i,j}$ and $T^{i,j}$. Thus, the total number of envy-pairs is at most $n^2 + m$ w.r.t $M$.
\end{proof}

\begin{lem}\label{lem:no-ep}
If $G$ is a ``no" instance of the Vertex Cover problem, then any solution of $I$ has at least $n^2 + m + 1$ envy-pairs.
\end{lem}
\begin{proof}
It is equivalent to prove that if $I$ admits a feasible matching $M$ with less than $n^2 + m + 1$ envy-pairs, then $G$ has a vertex cover of size at most $K$. $M$ must be a perfect matching and must be a one-to-one correspondence between $C \cup F$ and $V$ in order to be a feasible matching. Since all $V$ are matched to $C \cup F$, for each gadget $g_{i,j} = (S^{i,j}, T^{i,j})$, by Lemma \ref{lem:perfect}, there are only two possibilities, $M^{i,j}_0$ and $M^{i,j}_1$ and either matching admits one envy-pair within each $g_{i,j}$. Totally we have $m$ envy-pairs between $S$ and $T$.

Suppose we choose $M^{i,j}_0$ for $g_{i,j}$. If $v_j$ is matched with a resident in $F$, there are $l$ envy-pairs between $v_j$ and $S^{i,j}$. Then we have $l+m = n^2 + m + 1$ envy-pairs, contradicting the assumption. Thus, $v_j$ must be matched with a resident in $C$. By the same argument, if $M^{i,j}_1$ is chosen, $v_i$ must be matched with a resident in $C$. Hence, for each edge $(v_i, v_j)$, either $v_i$ or $v_j$ is matched with a resident in $C$. It is obvious that the set of vertices matched with residents in $C$ is a vertex cover of size $K$. This completes the proof.
\end{proof}

Thus a polynomial-time algorithm for 0-1 Min-EP HRLQ would solve the Vertex Cover problem, implying P=NP.
\end{proof}

Note that the reduction above also implies NP-hardness of 0-1 Min-BP HRLQ because all envy-pairs are blocking pairs and the construction in Lemma \ref{lem:yes-ep} does not generate any wasteful pairs (non-envy blocking pairs). \cite{hamada2016hospitals} gives stronger results showing 0-1 Min-BP HRLQ is hard to approximate within the ratio of $(|H|+ |R|)^{1-\epsilon}$ for any positive constant $\epsilon$ even if all preference lists are complete. While, the reduction can not be used for 0-1 Min-EP HRLQ because there exists envy-free matchings when all preference lists are complete \cite{fragiadakis2016strategyproof} (or a weaker requirement such that all hospitals with positive lower quotas has complete preference lists over all residents \cite{krishnaa2020envy}) and a maximum-size envy-free matching can be found in linear time.

\subsection{A Simple Exponential-Time Algorithm}
Let $I$ be a given instance. Starting from $k = 1$, we guess a set $B$ of $k$ envy-pairs. There are at most $|E|^k$ choices of $B$. For each choice of $B$, we delete each $(r,h) \in B$. Let $I'$ be the modified instance. We apply Yokoi's algorithm to find an envy-free matching in $I'$. If the algorithm outputs an envy-free matching $M$, then it is the desired solution, otherwise, we proceed to the next guess. If we run out of all guess of $B$ for a fixed $k$, we increment $k$ by 1 and proceed as before until we find a desired solution.

\begin{thm}
There is an $O(|E|^{t+1})$-time exact algorithm for Min-EP HRLQ where $t$ is the number of envy-pairs in an optimal solution of a given instance.
\end{thm}

\begin{proof}
If the algorithm ends when $k < t$, then the output matching $M$ is an envy-free matching in $I'$ and the set $B$ can only introduce at most $k$ envy-pairs, contradicting that $t$ is minimum number of envy-pairs in any feasible matching of instance $I$. Consider the execution of the algorithm for $k = t$ and any optimal solution $M_{opt}$ and let our current guess $B$ contains exactly the $t$ envy-pairs of $M_{opt}$. Then it is easy to see that $M_{opt}$ is envy-free in $I'$ and satisfies all the lower quotas. Hence if we apply Yokoi's algorithm, we find a matching $M$ is envy-free and satisfies all the lower quotas. Thus $M$ has at most $t$ envy-pairs in the original instance $I$ and it is our desired solution.

For the time complexity, Yokoi's algorithm runs in $O(|E|)$ time and for each $k$, we apply Yokoi's algorithm at most $|E|^k$ times. Therefore, the total time is at most $\sum^{t}_{k=1} O(|E|^{k+1}) = O(|E|^{t+1})$.
\end{proof}

\section{Minimum-Envy-Resident HRLQ}
In this section, we consider the problem of minimizing the number of envy-residents in HRLQ. We prove a NP-hardness result for Min-ER HRLQ in the following theorem.

\begin{thm}
Min-ER HRLQ is NP-hard.
\end{thm}
\begin{proof}
We give a polynomial-time reduction from the NP-complete problem CLIQUE. In CLIQUE, we are given a graph $G=(V,E)$ and a positive integer $K \leq |V|$, and asked if $G$ contains a complete graph with $K$ vertices as a subgraph.

{\bf Reduction:} Given a graph $G=(V,E)$, and a positive integer $K \leq |V|$, which is an instance of the CLIQUE problem, we construct in instance $I$ of Min-ER HRLQ. Define $n = |V|$, $m = |E|$ and $t = n+1$. The set of residents is $R = C \cup F \cup E'$ and the set of hospital is $H = V \cup \{x\}$. Each set is defined as follows:

\begin{align*}
    C &= \{c_i ~|~ 1 \leq i \leq K\}\\
    F &= \{f_i ~|~ 1 \leq i \leq n-K\}\\
    E' &= \{e^k_{i,j} ~|~ (v_i, v_j) \in E, 1 \leq k \leq t\}\\
    V &= \{v_i ~|~ 1 \leq i \leq n\}\\
\end{align*}

Each hospital in $V$ has a quota $[1,1]$ and the hospital $x$ has a quota $[mt, mt]$.

The preference lists of residents and hospitals is shown as follows:

\begin{align*}
    c_i &: [[V]] && (1 \leq i \leq K)\\
    f_i &: [[V]] && (1 \leq i \leq n-K)\\
    e^k_{i,j} &: v_i ~~~ v_j ~~~ x && ((v_i, v_j) \in E, 1 \leq k \leq t)\\
    \\
    v_i &: [[C]] ~~ [[E'_i]] ~~ [[F]] && (1 \leq i \leq n)[1,1]\\
    x &: [[E']] && [mt,mt]
\end{align*}
, where $[[V]]$ are a fixed order of all hospitals in $V$ in an increasing order of indices, $[[C]]$ and $[[F]]$ are a fixed order of all the residents in $C$ and $F$, respectively, in an increasing order of indices, $[[E'_i]]$ is an arbitrary order of all the residents in $E'$ that is acceptable to $v_i$ as determined by the preference lists of residents. $[[E']]$ are a fixed order of all residents in $E'$ in an increasing order of indices.

\begin{lem}\label{lem:yes-er}
If $I$ is a ``yes" instance of CLIQUE, then there is a feasible matching of $I$ having at most $(m - {K \choose 2})t+n$ envy-residents.
\end{lem}
\begin{proof}
Suppose that $G$ has a clique $V_c$ of size $K$. We will construct a matching $M$ of $I$ from $V_c$. We assign all the residents in $C$ to the hospitals in $V_c$ in an arbitrary way and all the residents in $F$ to the hospitals in $V \backslash V_c$ in an arbitrary way. Further, we match all the residents in $E'$ to $\{x\}$. Clearly, $M$ is feasible. Since $V_c$ is a clique, $(v_i, v_j) \in E$ for any pair $v_i, v_j \in V_c (i \neq j)$. There are $t$ residents $e^k_{i,j} (1 \leq k \leq t)$ associated with the edge $(v_i,v_j)$. Each of these residents $e^k_{i,j}$ are assigned to the hospital $x$ in $M$, which is the worst in $e^k_{i,j}$'s preference list and the hospitals $v_i$ and $v_j$ are assigned to residents in $C$, better than $e^k_{i,j}$. Hence all these residents $e^k_{i,j} (v_i, v_j \in V_c, i \neq j)$ are not envy-residents. There are ${K \choose 2}t$ such residents and total number of residents is $mt + n$. Hence there are at most $(m - {K \choose 2})t+n$ envy-residents.
\end{proof}

\begin{lem}\label{lem:no-er}
If $I$ is a ``no" instance of CLIQUE, then any feasible matching of $I$ contains at least $(m - {K \choose 2}+ 1)t$ envy-residents.
\end{lem}
\begin{proof}
Suppose that there is a feasible matching $M$ of $I$ that contains less than $(m - {K \choose 2}+ 1)t$ envy-residents. We show that $G$ contains a clique of size $K$. Note that $M$ must match all the resident in $E'$ to $\{x\}$ because $\{x\}$ has a quota $[mt, mt]$ and is only acceptable to $E'$. Thus the hospitals $V$ must only match to the residents in $C \cup F$ and there is one-to-one correspondence between $V$ and $C \cup F$ in order for $M$ to be feasible. Define $V_c$ be the set of hospitals matched with $C$. Clearly $V_c = K$. We claim that $V_c$ is a clique.

The total number of residents is $mt + n$. Since we assume that there are less than $(m - {K \choose 2}+ 1)t$ envy-residents, there are more than $n + ({K \choose 2} - 1)t$ non-envy-residents (obviously an non-envy-resident is a resident that is not an envy-resident). Since $|C| + |F| = n$, there are more than $({K \choose 2}-1)t$ non-envy-residents in $E'$. Consider the following partition of $E'$ into $t$ subsets: $E'_k = \{e^k_{i,j} ~|~ (v_i,v_j) \in E\}$ for each $1 \leq k \leq t$. There must exists a $k$ such that $E'_k$ contains at least $K \choose 2$ non-envy-residents by Pigeonhole principle. In order for $e^k_{i,j}$ to be a non-envy-resident, we note that both $v_i$ and $v_j$ must match to some resident in $C$ because if $v_i$ or $v_j$ is matched to $F$, we have an envy-pair contains $e^k_{i,j}$. Thus any pair of vertices in $V_c$ causes such a non-envy-resident, implying that $V_c$ is a clique. 
\end{proof}

Because $t = n + 1$, we have $(m - {K \choose 2}+ 1)t > (m - {K \choose 2})t+n$. Hence by Lemma \ref{lem:yes-er} and Lemma \ref{lem:no-er}, Min-ER HRLQ is NP-hard.
\end{proof}

\section{Conclusion and Open Problems}
In this paper, we give NP hardness results of minimizing envy in terms of envy-pairs and envy-residents in the Hospitals/Resident problem with Lower Quota. Hamada et al. \cite{hamada2016hospitals} showed hardness of approximation for the problem of minimizing the number of blocking pairs among all feasible matchings. It would be interesting to show hardness of approximation for minimizing envy in the HRLQ problem.

\bibliographystyle{plain}
\bibliography{references}

\end{document}